\documentclass[nonacm,sigplan]{acmart}

\usepackage{xcolor}
\usepackage[english]{babel}
\usepackage{blindtext}
\usepackage{subcaption}

\renewcommand\footnotetextcopyrightpermission[1]{} 
\setcopyright{none}

\acmDOI{}

\acmISBN{}

\acmConference[]{}
\acmYear{}
\copyrightyear{}

\usepackage{zref-totpages}
\usepackage{amsmath}
\usepackage{xurl}
\usepackage{diagbox}

\usepackage{amssymb,amsfonts}
\usepackage{amsthm}
\usepackage[ruled,vlined,linesnumbered]{algorithm2e}
\usepackage{booktabs}  
\usepackage{hyperref}
\usepackage{booktabs}
\usepackage{subcaption}
\usepackage{multirow}
\usepackage{mathtools}
\usepackage{tikz} 
\usepackage{float}
\usepackage{textcomp}
\usepackage{xcolor}
\usepackage[inline]{enumitem}
\usepackage{pgf} 
\usepackage{multirow}
\usepackage[ruled]{algorithm2e}

\makeatletter
\newcommand\sfparagraph{\@startsection{subparagraph}{5}{0em}%
{-.5\baselineskip \@plus -2\p@ \@minus -.2\p@}%
{-3.5\p@}%
{\sffamily \bfseries}}

\newcommand\takeaway{\@startsection{subparagraph}{5}{0em}%
{-.5\baselineskip \@plus -2\p@ \@minus -.2\p@}%
{-3.5\p@}%
{\bfseries}}
\makeatother

\newtheorem{theorem}{Theorem}[section]
\newtheorem{definition}{Definition}

\usepackage{tcolorbox}
\tcbuselibrary{breakable}
\tcbuselibrary{skins}

\newtcolorbox{sidenote}{enhanced jigsaw,breakable,colback=gray!8,colframe=black,width=\linewidth, boxsep=0pt,  before skip=3pt, after skip=3pt, left=5pt, right=5pt, top=4pt, arc=6pt, bottom=4pt,boxrule=0.6pt,rightrule=1.4pt,bottomrule=1.4pt, pad at break=2pt, bottomrule at break=0pt, toprule at break=0pt}

\usepackage{lipsum}

\usepackage{circledsteps}
\DeclareRobustCommand{\textcircled}[1]{\Circled{#1}}

\hypersetup{pdfcontactemail={hanxinchi@sjtu.edu.cn}}
\makeatletter
\AtBeginDocument{\let\hyxmp@parse@acmart\relax}
\makeatother

\author{Xinchi Han}
\authornote{Both authors contributed equally to this research.}
\affiliation{%
  \institution{Shanghai Jiao Tong University}
  \city{Shanghai}
  \country{China}
  \email{hanxinchi@sjtu.edu.cn}
}

\author{Weihao Jiang}
\authornotemark[1]
\affiliation{%
  \institution{Shanghai Jiao Tong University}
  \city{Shanghai}
  \country{China}
  \email{weihao.jiang@sjtu.edu.cn}
}

\author{Yongxi Lv}
\affiliation{%
  \institution{Shanghai Jiao Tong University}
  \city{Shanghai}
  \country{China}
  \email{shjdblgklyx2435@sjtu.edu.cn}
}

\author{Yingming Mao}
\affiliation{%
  \institution{Xi'an Jiaotong University}
  \city{Xi'an}
  \country{China}
  \email{mao1234@stu.xjtu.edu.cn}
}

\author{Changbo Wu}
\affiliation{%
  \institution{University of Science and Technology of China}
  \city{Hefei}
  \country{China}
  \email{wuchangbo@mail.ustc.edu.cn}
}

\author{Liu yike}
\affiliation{%
  \institution{Xi'an Jiaotong University}
  \city{Xi'an}
  \country{China}
  \email{cn-lyk@stu.xjtu.edu.cn}
}

\author{Ximeng Liu}
\affiliation{%
  \institution{Shanghai Jiao Tong University}
  \city{Shanghai}
  \country{China}
  \email{liuximeng@sjtu.edu.cn}
}

\author{Haoyuan Liu}
\affiliation{%
  \institution{Shanghai Jiao Tong University}
  \city{Shanghai}
  \country{China}
  \email{cse-lhy@sjtu.edu.cn}
}

\author{Kairui Zhou}
\affiliation{%
  \institution{Shanghai Jiao Tong University}
  \city{Shanghai}
  \country{China}
  \email{kairuizhou@sjtu.edu.cn}
}

\author{Shuyuan Zhang}
\affiliation{%
  \institution{Shanghai Jiao Tong University}
  \city{Shanghai}
  \country{China}
  \email{zhang-shuyuan@sjtu.edu.cn}
}

\author{Peirui Cao}
\affiliation{%
  \institution{State Key Laboratory of Novel Software Technology, Nanjing University}
  \city{Nanjing}
  \country{China}
  \email{caopeirui@nju.edu.cn}
}

\author{Xun Nie}
\affiliation{%
  \institution{Xidian University}
  \city{Xi'an}
  \country{China}
  \email{xunnie@stu.xidian.edu.cn}
}

\author{Xuanyu Wang}
\affiliation{%
  \institution{LightDance}
  \email{wangxuanyu@lightdance.hk}
}

\author{Jiachen Wang}
\affiliation{%
  \institution{LightDance}
  \email{wangjiachen@lightdance.hk}
}

\author{Yujun Liu}
\affiliation{%
  \institution{LightDance}
  \email{liuyujun@lightdance.sg}
}

\author{Xinbing Wang}
\affiliation{%
  \institution{Shanghai Jiao Tong University}
  \city{Shanghai}
  \country{China}
  \email{xwang8@sjtu.edu.cn}
}

\author{Shizhen Zhao}
\affiliation{%
  \institution{Shanghai Jiao Tong University}
  \city{Shanghai}
  \country{China}
  \email{shizhenzhao@sjtu.edu.cn}
}

\begin{document}

\title{On the Design of Physical Topology in OCS-based Clusters}
\renewcommand{\shortauthors}{X. Han et al.}
\renewcommand{\shorttitle}{}
\begin{abstract}

Optical Circuit Switches (OCSes) are increasingly deployed in data center networks and AI clusters. By wiring the electrical ports of GPUs and electrical switches to the OCS ports, the physical interconnection pattern defines a \emph{physical topology}.  \textbf{Topology Engineering (ToE)} seeks an OCS configuration that optimizes a traffic objective given traffic demands and physical constraints, but solving this joint problem is computationally expensive. A common approach is to \emph{decouple} the problem into two steps: first, compute a \emph{logical topology} that satisfies the traffic demands without considering the physical topology constraints and then map it onto the physical topology to obtain the OCS configuration. This, however, introduces a quality loss issue: the resulting logical topology sometimes cannot be realized exactly under the given physical topology.

To address this problem, we show that it stems from a property missing in the physical topology. We propose \textbf{Decoupling Optimality} and the necessary and sufficient conditions under which a decoupled method will not lose quality. Examining existing physical topologies against these conditions reveals that they either fail them, or satisfy them at the cost of cluster scale or deployment generality. Our proposed \textbf{Cross Wiring} meets all the conditions, without sacrificing cluster scale or adding devices. On a 128-NPU testbed, the lack of \emph{Decoupling Optimality} increases training iteration time by up to 39.5\%, and a microbenchmark validates the importance of \emph{Decoupling Optimality} in both solution quality and solving speed. A testbed built from two MEMS OCSes validates the engineering feasibility of Cross Wiring.


\end{abstract}

\settopmatter{printfolios=true}
\maketitle

\begin{sloppypar}

\section{Introduction}


Optical circuit switches (OCSes) are increasingly deployed in data center networks and large-scale AI clusters \cite{poutievski2022jupiter, 2023TPU,Patronas:25,GoogleA3Supercomputers}. In an OCS cluster, end devices such as servers, GPUs, and electrical switches are grouped into \textbf{EPS Group}s (\textbf{EG}s), e.g., a rack \cite{2023TPU} or a pod \cite{poutievski2022jupiter} (fig.~\ref{fig:eg}). When the cluster is built, the ports of each EG are cabled to OCS ports; this fixed cabling pattern forms the \textbf{Physical Topology} (fig.~\ref{fig:physical}). At runtime, an OCS can reconfigure how its ingress ports pair with its egress ports; each such setting is an \textbf{OCS Configuration} (fig.~\ref{fig:configuration}), and it determines how many links connect every pair of EGs, i.e., the \textbf{Logical Topology} exposed to workloads (fig.~\ref{fig:logical}). Given the communication demands of a workload, \textbf{Topology Engineering} (ToE) computes the OCS configuration that optimizes a traffic-related objective subject to the constraints of the physical topology.

\emph{Joint} methods~\cite{farrington2010helios,wang2022topoopt} solve ToE directly, but the joint search space grows combinatorially with the cluster scale, and the wiring constraints must be reformulated for every new deployment. \emph{Decoupled} methods~\cite{Liao_2025,2023TPU} therefore decompose ToE into two stages: first, compute a logical topology that matches the traffic demands and then configure the OCSes to realize it. However, the logical topology computed in the first stage may be unrealizable on the physical topology; the deployed configuration only approximates it, degrading performance, and sometimes failing connections. We call the gap between the target and the realized logical topology \textbf{Quality Loss}. Existing work treats that Quality Loss as an inherent price of decoupling, and compensates for it with ad-hoc fixes, robust approximations, and greedy/heuristic searches~\cite{poutievski2022jupiter,2023TPU,Liao_2025,10892202} that carefully balance solution quality against solving efficiency.

\begin{figure}[tb]
     \begin{subfigure}[b]{0.64\linewidth}
     \centering
     \includegraphics[height=0.70\linewidth]{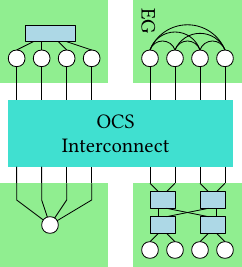}
     \Description{The image draws 4 types of EGs connecting to an interconnect square. The interconnect is a pure block representing any type of OCSes (compared to Physical topology)}
    \caption{EGs and OCS Interconnect}\label{fig:eg}
     \end{subfigure}
     \hfill
     \begin{subfigure}[b]{0.35\linewidth}
     \centering
     \includegraphics[width=0.9\linewidth]{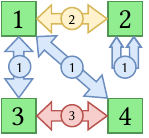}
     \Description{The image draws a vertex-edge representation of following matrix, where weight of edges is the value.}
     \vspace{-0.7ex}
     {$L =$ \LARGE $\begin{bsmallmatrix}
    0 & 2 & 1 & 1 \\
    2 & 2 & 0 & 0 \\
    1 & 0 & 0 & 3 \\
    1 & 0 & 3 & 0
    \end{bsmallmatrix}$
     }
    \caption{Logical Topology}
    \label{fig:logical}
     \end{subfigure}

     \begin{subfigure}[b]{0.45\linewidth}
     \centering
     \includegraphics[scale=0.6]{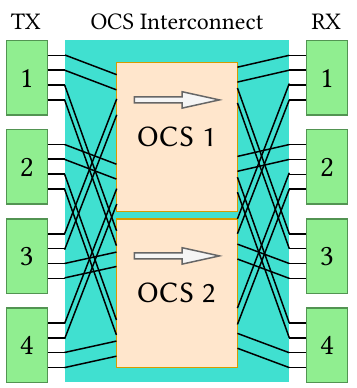}
    \caption{Physical Topology}\label{fig:physical}
    \Description{The image draws detailed connection inside the OCS, connecting 4 EGs with 2 OCSes in a 2-Uniform wiring way.}
     \end{subfigure}
     \hfill
     \begin{subfigure}[b]{0.45\linewidth}
     \centering
     \includegraphics[scale=0.6]{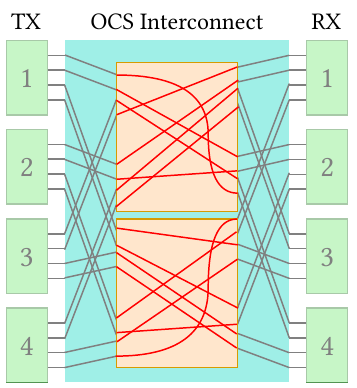}
     \Description{The image draws OCS configurations as red wire, demonstrating configurations on the overlay of physical topology.}
    \caption{OCS Configuration}
    \label{fig:configuration}
     \end{subfigure}
     \caption{Components of ToE}
\end{figure}

We instead question an assumption that existing work leaves implicit: \textit{Under what physical topology can ToE be decoupled without Quality Loss?} Our answer is the characterization \textbf{Decoupling Optimality}: it holds when, for every workload, the two-stage solution does not degrade in optimization objective compared to joint method. We establish the necessary and sufficient conditions for Decoupling Optimality (\S\ref{sec:decoup}), leading to a requirement that the physical topology be \emph{logical-topology compatible}, i.e. able to deliver all logical topologies through OCS reconfiguration.

Revisiting existing designs against this condition reveals the root cause of quality loss. The prevailing \textsf{Uniform} wiring, which has been regarded as optimal in previous works~\cite{zhao2021understanding}, violates the compatibility thus not Decoupling Optimal(\S\ref{sec:uniform}). Wirings that restore this compatibility do so only at a cost: \textsf{Single} limits the cluster to the port count of a single OCS (\S\ref{sec:suboptimal}); \textsf{2-Uniform}~\cite{zu2024resiliency,Patronas:25} halves the attainable cluster scale (\S\ref{sec:2-uniform}); and \textsf{2S-Uniform}~\cite{poutievski2022jupiter} depends on additional optical components (\S\ref{App:slim-dual}).

Guided by this requirement, we propose \textbf{Cross Wiring} (CW), a physical topology that satisfies logical-topology compatibility and admits a polynomial-time OCS configuration algorithm (\S\ref{sec:phoscore}). It retains the same scale as \textsf{Uniform} without additional hardware, and an existing \textsf{Uniform} cluster can be converted by rewiring its fiber connections (\S\ref{sec:case_study}).

On a 128-NPU testbed, \textsf{Uniform} increases training iteration time by up to 39.5\% over Cross Wiring due to quality loss, and trace-driven large-scale simulations show the same effect. In fast-switching scenarios \cite{Liao_2025}, CW resolves ToE's dilemma between solution quality and solving speed. In microbenchmarks, our OCS configuration algorithm on CW exactly realizes the logical topology produced by optimizers such as ATRO~\cite{mao2025atro}, with maximum link utilization (MLU) matching the ATRO bound. On \textsf{Uniform}, the joint method based on integer linear programming (ILP) is $\sim$8400$\times$ slower than our algorithm, while the ILP-based decoupled method yields MLU 499.3\% higher than ours.

In summary, our contributions are as follows:
\begin{itemize}
    \item We identify the root cause of Quality Loss in decoupled ToE by formalizing Decoupling Optimality and establishing its necessary and sufficient conditions.

    \item We revisit existing physical topologies and find that they either violate logical-topology compatibility or satisfy it at the cost of cluster scale or extra optical components.

    \item We design Cross Wiring, a logical-topology-compatible physical topology that, together with a fast OCS configuration algorithm, achieves Decoupling Optimality at the same scale as \textsf{Uniform} with in-place deployability.

    \item We demonstrate the value of Decoupling Optimality on a 128-NPU testbed and in large-scale trace-driven simulations, and validate the deployability of Cross Wiring with a testbed built from two MEMS OCSes.
\end{itemize}

\textit{This work does not raise any ethical concerns.}

\section{Revisiting Topology Engineering}

\subsection{Definition of ToE Problem}\label{sec:model_def}



In this section, we formalize the ToE problem. ToE selects an OCS configuration under the physical-topology constraints to maximize a traffic-related utility function $U$ defined over logical topologies. The utility adopts different forms in different scenarios. For instance, minimizing the maximum link utilization (MLU) can be formulated as a utility maximization problem, while the optimization variables and constraints remain identical regardless of $U$. 

Existing work solves ToE in two forms. For a workload $W\in\mathbb{W}$, the \textit{joint} method $J\colon\mathbb{W} \to \mathbb{O}$ optimizes utility $U$ directly over the set of OCS configurations $\mathbb{O}$ admitted by the physical topology: $O_{\text{Joint}}=J(W)$. The \textit{decoupled} method splits ToE into two stages: the logical-topology solver $L\colon\mathbb{W}\to\mathbb{L}$ first computes $L(W)$ by optimizing $U$ over the set of all logical topologies $\mathbb{L}$; the configuration solver $C\colon\mathbb{L}\to\mathbb{O}$ then computes the deployed configuration $O_{\text{decoupled}} = C(L(W))$.



\sfparagraph{Components.} Each participant of OCS is a group of hosts or electrical switches that attaches to the OCSes through a bundle of equivalent fiber links and communicates with other participants solely through the OCSes; we call each participant an \textbf{EPS Group} (\textbf{EG}). For the ToE model, an EG is simply a set of equivalent OCS-facing ports, and its internal organization is immaterial: an EG can be a rack or a pod with an internal mesh (e.g., a Dragonfly fabric~\cite{7877093}), Clos~\cite{GoogleA3Supercomputers}, or multi-rail fabric~\cite{10.1145/3651890.3672265,10.1145/3789240.3829187}, or a single node whose multiple ports directly face the OCSes~\cite{wang2022topoopt,Liao_2025} (fig.~\ref{fig:eg}). A cluster can deploy multiple OCSes, which together form an \emph{interconnect layer} to scale out the cluster.






The \textbf{Physical Topology} is fixed by cabling when the cluster is built: it defines how EG ports are cabled to the OCSes and remains unchanged after deployment; it is therefore a given input to ToE. The ports of an EG may be spread across multiple OCSes. A port on the EG side carries both directions of a link, whereas an OCS is a unidirectional cross-connect: its ports are divided into an \emph{N} side and an \emph{S} side, and an optical path can be established only between an N port and an S port, never between two ports on the same side. From the OCS's perspective, each EG port is consequently split into a transmit (TX) half-port and a receive (RX) half-port: all TX half-ports terminate on one side and all RX half-ports on the other. The physical topology is thus exactly the mapping from every port on EGs to an N or S port on some OCS (fig.~\ref{fig:physical}).

A \textbf{Logical Topology} is an abstract description of how many links connect every pair of EGs in each direction, commonly represented as a matrix (fig.~\ref{fig:logical}); as a description of demands, it exists independently of any particular OCS configuration. At runtime, each OCS pairs its N ports one-to-one with its S ports, and the pairings on all OCSes together constitute an \textbf{OCS configuration} (fig.~\ref{fig:configuration}). 







Whether a configuration realizes a target logical topology depends on three constraints.

1) \textbf{Logical-topology constraint.} The links established by the configuration must match the target logical topology: between every pair of EGs, the number of links in each direction must equal the demanded number.

2) \textbf{Physical-topology constraint.} Each TX/RX port carries at most one light path, and an optical path can only cross-connect an N port with an S port on the same OCS, with the reachable pairs fixed by the cabling.

3) \textbf{Symmetry constraint.} The two directions of every inter-EG link must use the same pair of EG ports: if the TX of port A on one EG connects to the RX of port B on one EG, then the RX of port A must also connect to the TX of port B.


Joint method also takes constraint 2 and 3, since they are imposed by the cluster connection and hardware. Solving three constraints yields an integer linear programming (ILP) problem, we formulate it in the technical report.

\sfparagraph{Symmetry constraint.} Among three constraints, symmetry constraint has been frequently ignored or over-simplified among three constraints, though it plays a very important role in EPSes. Link layer protocol such as 802.3 Ethernet, FibreChannel, and InfiniBand requires duplex connection over link \cite{EthernetAddressResolution1982,palkert2018dac}. Several physical layer capabilities used in optical connections like auto-negotiation also relies on connection to both sides. If only a single direction connection, or two direction connection over different ports is established, the link layer will render the connection down. Much prior work has focused on the control-plane aspects of OCS, EPS, and end hosts \cite{wang2010cthrough,porter2013microsecond,mellette2020opera,zhao2021understanding,6490069} but omitting the constraint, resulting in bandwidth loss and connection failures when configuration is actually deployed on clusters.

\subsection{Decoupling Optimality}\label{sec:decoup}

Compared to Joint method, decoupled solving may fall into a situation
when the two stages' constraints conflict and no exact configuration can be found. \textit{Relaxation} is therefore employed to deliver a ToE result with best effort, by converting the \textit{logical-topology constraint} into a soft constraint or an optimization objective. The actual formed logical topology will then deviate, and resulting in optimization \textbf{Quality Loss} in general cases, since traffic-related objectives are out of scope in this decoupled stage. Quality loss will further cause performance degradation and can even bring links down. We show such a loss on an example optimizing MLU in the technical report. 


We don't believe the quality loss is an inevitable trade-off to solve in a decoupled way. Consider the OCS configuration step, determined by the algorithm and physical topology: $\langle C, \mathcal P \rangle$. We formalize the quality loss:

Given a pair $\langle C, \mathcal P \rangle$, no quality loss means for arbitrary workload $W$ and utility function $U$:
\[
U\big(C(L(W))\big) \;\ge\; U\big(J(W)\big).
\]

Since physical topologies dominate how configurations are searched, we define \textbf{Decoupling Optimality} to characterize the physical topology that quality loss can be avoided:

\begin{definition}[Decoupling Optimality]\label{th:lossless}
A physical topology $\mathcal P$ has decoupling optimality, if and only if there exists a configuration algorithm $C$ s.t. for arbitrary utility function $U$ and workload $W$, the utility of decoupled method will not degrade:
$$
\exists C \text{ s.t. } U\big(C(L(W))\big) \;\ge\; U\big(J(W)\big) \text{ over } \langle C, \mathcal P\rangle. 
$$
\end{definition}

Decoupling Optimality is equivalent to quality lossless method existence. To understand when and how relaxation and quality loss will happen, we investigate both method's domain and range. Consider a \textbf{feasible set}, which represents all logical topology can be supported by OCS configurations with given physical topology. 

\begin{definition}[Feasible Set]\label{def:feasible_set}
We denote a function $\mathrm{Logical}: \mathbb{O} \mapsto \mathbb{L}$ mapping an OCS configuration to the logical topology it realizes. The feasible set is
$$
F = \{\,\mathrm{Logical}(O) : O \in \mathbb{O}\,\} \subseteq \mathbb{L}
$$
\end{definition}

\begin{figure}[htbp]
    \centering
    \includegraphics{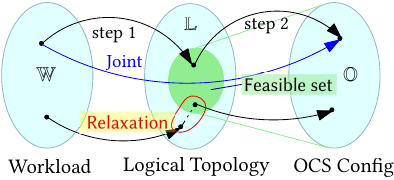}
    \caption{Set Representation of Two Methods}
    \label{fig:decouple}
\end{figure}

As illustrated in fig.~\ref{fig:decouple}, the upper arrow in fig.~\ref{fig:decouple} shows decoupling can preserve the quality when logical topology solver's result falls into the feasible set. While for the lower one, the result falls outside, forcing relaxation for any OCS configuration solver $C$, and can cause quality loss. A physical topology must have a full feasible set to achieve no quality loss:

\begin{proposition}[Decoupling Optimality]\label{def:DecouplingOptimality}
For a physical topology $\mathcal P$, Decoupling Optimality holds if and only if
$$
F = \mathbb{L}.
$$
\end{proposition}

When decoupling optimality holds, we have $F = \mathbb L$. With a brute-force algorithm enumerating $\mathrm{Logical}(O)$, the OCS configuration method $\langle C_{\text{BF}}, \mathcal P \rangle$ requires no relaxation and always yields exact results. This follows directly from Th. \ref{th:decoupling_optimality}, the necessary and sufficient  condition for zero quality loss in decoupled schemes, which we further elaborate in \S\ref{sec:phoscore}. 

While a brute-force construction is theoretically valid, it is computationally impractical. Before discussing how to find an efficient algorithm $C$, we here discuss how to reach decoupling optimality for a physical topology.

\sfparagraph{Decoupling Optimal Physical Topologies.} 
Decoupled optimality requires that for all logical topology, there exists an OCS configuration to support it over the physical topology. Though it may appear self-evident and is often taken for granted in the ToE literature, decoupling optimality does not hold for every physical topology. \emph{Single-link Uniform Wiring} (\textsf{Uniform}, cf. \S\ref{sec:uniform} for details), which is widely adopted by Helios \cite{farrington2010helios}, Opera \cite{mellette2020opera}, RotorNet \cite{mellette2017RotorNet}, TROD \cite{9651977}, TopoOpt \cite{wang2022topoopt}, and others~\cite{10892202,dong2025risk,zhao2021understanding}, is an example of a physical topology that does not satisfy this condition.

\begin{figure}[htbp]
\centering
    \begin{subfigure}[b]{0.44\linewidth}
    \centering
     \includegraphics[height=65pt]{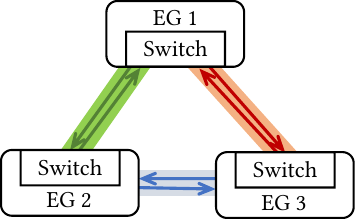}
     \caption{Logical Topology}
     \label{fig:l2_logical_topology}
     \end{subfigure}
     \begin{subfigure}[b]{0.55\linewidth}
     \centering
     \includegraphics[height=65pt]{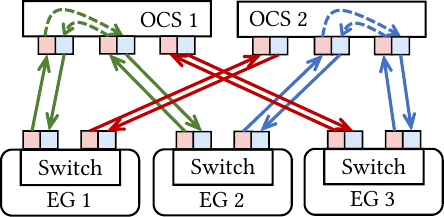}
     \caption{Implementation Failure}
     \label{fig:illu_l2_no_solution}
     \end{subfigure}
    \caption{\textsf{Uniform} is not Decouple Optimal}
    \label{fig:physical topology design challenge}
\end{figure}

We demonstrate this issue with a minimal example in fig.~\ref{fig:physical topology design challenge}. Consider a 3-EG cluster in which each EG has 2 ports and connects to 2 OCSes. The target logical topology (fig.~\ref{fig:l2_logical_topology}) requires every EG to connect to every other EG with a single link, while no OCS configuration algorithm can exactly realize it.

Given OCS configuration algorithm $C$, we quantify the resulting quality loss with the \emph{Logical Topology Compatibility Rate} (LTCR):

\begin{definition}[LTCR]
For OCS configuration algorithm $C$ on input $L$, Logical Topology Compatibility Rate $\mathrm{LTCR}$ is defined as:
\begin{gather}\label{eq:LTCR}
\mathrm{LTCR}^{C}(L) \;=\; \Delta\big(L,\, \mathrm{Logical}\big(C(L)\big)\big), \\ \nonumber 
\text{where }\Delta(L, L'):= \frac{1}{|L[i,j]>0|}\sum_{\substack{i,j\\L[i,j]>0}}\min\left\{\frac{L'[i,j]}{L[i,j]},\,1\right\}.
\end{gather}
\noindent $|L[i,j]>0|$ denotes the number of nonzero entries of $L$. 
\end{definition}

Beyond decoupling loss, the \textit{efficiency} and \textit{scalability} of a ToE algorithm are also key design considerations for physical topologies and ToE algorithms alike. We discuss them in \S\ref{sec:uniform} with \textsf{Uniform} and in \S\ref{sec:2-uniform} with \textsf{2-Uniform}.

\section{From Single OCS to Uniform Wiring}\label{sec:suboptimal}

Before discussing a new topology design, we first walk through the physical topology evolution of OCS clusters with the decoupling optimality perspective.

\subsection{Single OCS (\textnormal{\textsf{Single}})} 

Using a single OCS to interconnect all EGs (\textsf{Single}) is a naive strategy for achieving decoupling optimality. Various ToE algorithms, especially decoupled methods focusing on the logical topology side are constructed upon, or assumed a single OCS topology.~\cite{0On,Liao_2025,6490069}.

Exact configuration of single OCS is trivial with a $O(N^2)$ scan-and-set algorithm, which is commonly used as an ideal physical topology. However, with limited connectivity, since each EG must connect all ports to the OCS, the scalability is heavily reduced to maximum EGs by OCS port count. This barely provides the required connectivity due to OCS hardware limitations. So only used in small clusters or specialized local connectivity like \textit{MixNet}~\cite{Liao_2025}.

\subsection{Single-link Uniform Wiring (\textnormal{\textsf{Uniform}})} \label{sec:uniform}

An instant scaling method is using multiple OCSes, and connect EGs to OCSes in the time-tested Clos-shape method, which leads us to Single-Link Uniform Wiring. This topology evenly distributes OCS-facing links across all OCSes. With a Clos-shape providing non-blocking wiring, \textsf{Uniform} has long been regarded as optimal~\cite{zhao2021understanding}, and widely deployed in SoTA OCS-based clusters~\cite{zhao2021understanding,9651977,cao2023threshold,mellette2020opera,porter2013microsecond,mellette2017RotorNet,10892202}.

To adapt existing ToE algorithms based on \textsf{Single}, the decoupled approach was introduced. As multi-OCS deployments introduce additional constraints, several ILP-based methods have been proposed to compute OCS configuration schemes, under the long-standing assumption that decoupling optimality holds. However, the \textbf{symmetry constraint} fundamentally invalidates this presumed optimality, and violating symmetry constraint resulting in degradations and failures.

We now know that compromised decoupling optimality is the culprit, as \textsf{Uniform} cannot support arbitrary logical topologies, leading to nearly inevitable quality loss. Such loss in decoupling optimality not only introduces optimization object loss, but also heavily affects the actual performance. To quantify such performance impact, we evaluated ML tasks defined in \S\ref{section:testbed} on a 128-NPU cluster employing the \textsf{Uniform} topology. Results in fig. \ref{fig:uniform-no-do} show that this decoupling optimality loss can increase training iteration time by up to 39.5\%.

\begin{figure}[!b]
    \begin{subfigure}[b]{1.0\linewidth}
     \centering
     \includegraphics[width=\linewidth]{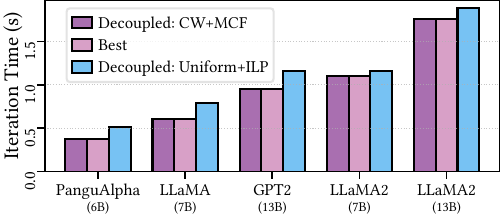}
     \vspace{-3.5ex}
     
     \emph{\footnotesize Cf. \S\ref{exp:given_logo} for Compared Methods}
    \vspace{-1.5ex}
     
    \caption{Quality loss harms training throughput}
    \label{fig:logical_topo_contention}
     \end{subfigure}
     \begin{subfigure}[b]{0.49\linewidth}
     \includegraphics[scale=1]{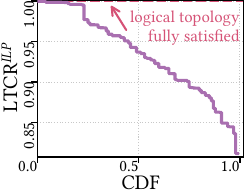}
    \caption{Degraded LTCR}
    \label{fig:Realization_Rate}
     \end{subfigure}
     \hfill
     \begin{subfigure}[b]{0.49\linewidth}
     \includegraphics[scale=1]{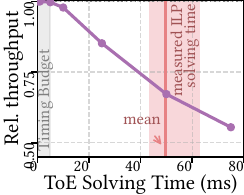}
    \caption{Slow ToE harms performance}
    \label{fig:toe_deadline}
     \end{subfigure}
     \caption{\textsf{Uniform} Lost Decouple Optimality}\label{fig:uniform-no-do}
\end{figure}





Besides, the high computation complexity makes \textsf{Uniform} less competitive. Before further elaboration, here we demonstrate why such complexity matters, especially with LLM training.

\begin{sidenote}
    \sfparagraph{Computation Efficiency.}

For a long time, OCS reconfiguration latency was high enough to dominate the reconfiguration life cycle~\cite{poutievski2022jupiter}. ToE computation time was treated as secondary and algorithmic designs focused primarily on solution quality. However, this premise no longer holds. Recent advances in OCS hardware~\cite{TODO_neye,ipronics_web,11573381} enables fast reconfiguration, with systems~\cite{Liao_2025,10.1145/3789240.3829187,ma2025railx} pushing reconfiguration to the collective-communication or operator level, requiring timely reconfiguration to sustain high throughput (cf. \S\ref{sec:layer_eval}).



We have measured the ToE+reconfigure delay timing budget of an operator-level reconfiguration scenario, using MixNet copilot \cite{Liao_2025} on an H800 testbed as an example (details in the technical report). Fig \ref{fig:toe_deadline} shows that a total delay exceeding 5ms will start to affect the model efficiency. When total delay increases, the pipeline stalls for configuration: a 10\,ms delay increases $T_{\text{iter}}$ by 2.2\%, and 50\,ms by 48.8\%.
\end{sidenote}

The computation complexity of \textsf{Uniform} is high. As shown in fig.~\ref{fig:toe_deadline}, our measurements show that ToE solving under \textsf{Uniform} takes $43.71\sim62.79$\,ms with an ILP solver. In fast reconfiguration scenario, quality and efficiency are no longer trade-offs, but a dilemma since sacrificing either would contribute to throughput and performance loss.

To relieve the timing pressure, several greedy and approximation algorithms emerge to compensate the gap~\cite{0On,Liao_2025}. However, exact solution of OCS configuration over \textsf{Uniform} is NP-Complete, and the time cannot be easily shortened without massive degradation to solution quality.

\begin{theorem}\label{theorem:uniform_suboptimal}
Under the Single-link Uniform Wiring physical topology (\textsf{Uniform}), the OCS configuration problem is NP-Complete.
\end{theorem}

The theorem is proved in the technical report, by reducing the \textsc{Multigraph Coloring} problem \cite{BARNOY2000422} to the ILP model.

\section{Decoupling-Optimal Physical Topology}\label{sec:phoscore}

\subsection{OCS Configuration Algorithm Completeness.}

When OCS configuration algorithm $C$ cannot find an exact solution when it exists, even decoupling optimal topologies are still in risk of quality loss. Denote the actual exact solution domain of $C$ as $F_C$:
$$
F_C = \{L \in \mathbb L: \mathrm{Logical}(C(L)) = L\} \subseteq F.
$$
The OCS configuration algorithm must deliver $F_C = F$, to become an algorithm specified in Def. \ref{th:lossless} and avoid quality loss. We denote such property as \textbf{completeness}.

Completeness ensures no relaxation will happen when exact solution exists. For example, exhaustive DFS and branch-and-cut ILP solvers are complete algorithms, whereas approximation algorithms, such as random search, greedy heuristics, and ILP solving under a time bound are not complete. 

With algorithm constrained by the completeness and physical topology constrained by the decoupling optimality, we can extend Prop. \ref{def:DecouplingOptimality} to deliver the necessary and sufficient condition to avoid decoupling quality loss:

\begin{theorem}[Necessary and Sufficient Condition to avoid Decoupling Loss]\label{th:decoupling_optimality}
The quality loss is eliminated if and only if the exact solution set covers every logical topology:
$$F_C = \mathbb{L}.$$
Which can be further decomposed into (1) OCS configuration algorithm Complete, and (2) Physical topology Decoupling Optimal:
$$
F_C = \mathbb{L} \Longleftrightarrow F_C = F \land F = \mathbb{L}
$$
\end{theorem}

We prove the theorem and Prop. \ref{def:DecouplingOptimality} in the technical report.

\subsection{A Polynomial Decoupling-Optimal OCS Reconfiguration Algorithm}\label{sec:polynomial}

Given the necessary and sufficient conditions for a physical topology to eliminate quality loss: \textcircled{1} it must be decoupling optimal and \textcircled{2} admit a complete configuration algorithm. \textsf{Uniform} physical topology fails \textcircled{1}, and high complexity in solving practically prevents \textcircled{2}. To find a decoupling optimal topology, we first characterize the structural condition under which a configuration algorithm is complete, and then turn the condition into a design requirement on the physical topology. Then, the physical topology and algorithm obtained along this path can satisfy the two conditions at the same time, guaranteeing decoupled solving quality.

We first examine the optimization structure of the OCS configuration problem. For OCS configuration on\textsf{Uniform}, the problem is NP-complete. However, if we remove the \textit{symmetry constraint} from the constraints, the remaining closely resembles the Integer Matrix Decomposition problem: By splitting the logical topology matrix evenly into one sub-topology per OCS, the sub-problem reduces to simple matrix mapping used in \textsf{Single} OCS allocation, which is known to be complete, and can be solved in $O\left(N^2\right)$ time. The logical topology splitting can be further reduced to a minimum-cost flow (MCF) problem~\cite{zhang2023reducing,zhao2019minimal}, which is also proven complete, and solvable in polynomial time.


The symmetry constraint is difficult to embed into the optimization model, since it is a per-port-pair pairing constraint: every link from EG $i$ to EG $j$ must use exactly the same pair of ports as the reverse link from $j$ to $i$, the decision variables must be considered in pairs and cross-variable constraints is introduced. MCF-based polynomial solvers handle only submodular constraints and node-wise coupling cannot be expressed. Doing so turns it into a general combinatorial problem over graph, which is NP-hard.

\sfparagraph{Absorbing the symmetry constraint.}
Rather than placing the symmetry constraint inside the optimization model, we satisfy it at the level of problem construction. Instead of searching all configurations for a symmetric one, we reformulate the problem so that every feasible solution is symmetric by construction. With symmetry constraint, every logical topology matrix $L$ is necessarily a symmetric integer matrix, $L_{ij}=L_{ji}$. We can therefore decompose $L$ into an integer matrix and its transpose, $L=A+A^{\mathrm T}$, where $A$ itself can be an arbitrary matrix. As a sub-logical-topology, $A$ carries the links in one direction and $A^{\mathrm T}$ carries all reverse links. The directional pairing is then settled at the matrix level, and each of $A$ and $A^{\mathrm T}$ faces an ordinary capacity-constrained problem with no symmetry constraint.

This decomposition is not a trivial algebraic rewrite. The identity $L=L/2+L^{\mathrm T}/2$ holds for any symmetric matrix, but no longer ensures integer, which does not correspond to number of links; its row and column sums are also not guaranteed to be balanced, so the link count assigned to each OCS need not match its port capacity. The real question is whether an integral $A$ with balanced row and column sums always exists for every symmetric integer matrix. The following theorem answers this in the affirmative.

\begin{theorem}[Symmetric Integer Matrix Decomposition]\label{lem:matrix_decomp}
For any symmetric integer matrix $L$, there exists an integer matrix $A$ such that $L=A+A^{\mathrm T}$ and, for every $i,j\in[P]$,
\begin{align*}
\Big\lfloor\tfrac12\textstyle\sum_{j\in[P]} L_{ij}\Big\rfloor \le \sum_{j\in[P]} A_{ij} \le \Big\lceil\tfrac12\textstyle\sum_{j\in[P]} L_{ij}\Big\rceil,\\
\Big\lfloor\tfrac12\textstyle\sum_{i\in[P]} L_{ij}\Big\rfloor \le \sum_{i\in[P]} A_{ij} \le \Big\lceil\tfrac12\textstyle\sum_{i\in[P]} L_{ij}\Big\rceil.
\end{align*}
\end{theorem}

The two bounds keep every row and column sum of $A$ between the floors and ceilings of half the corresponding sums of $L$, reserving half of the transmit and receive capacity for each direction.

The theorem is constructive: We reduce the decomposition to a bounded MCF: link pairs act as supply nodes and EGs as demand nodes, while arc bounds jointly encode integrality, the pairing relation $A_{ij}+A_{ji}=L_{ij}$, and the row/column-sum bounds. By the integrality property of MCF: A feasible real-valued flow implies a feasible integer-valued flow~\cite{wolsey2020integer}; the existence of trivial real solution $A_{ij}=L_{ij}/2$ guarantees the existence and completeness (the full proof is in the technical report). On top of this reduction, we provide three exact solvers: solving the MCF directly with cost scaling costs $O(P^6\log P)$, a max-flow-based algorithm costs $O(P^3)$, and an Euler-tour-based algorithm costs $O(P^2)$. To adapt to varying problem scales, our implementation runs all 3 methods in parallel and returns the fastest result. The microbenchmark (\S\ref{sec:evn_log}) reports average overhead of 23.8\,$\mu$s, 74.5\,$\mu$s, 286.3\,$\mu$s, and 1.49\,ms for this decomposition step at the 2k, 4k, 8k, and 16k scales, respectively.

\subsection{Mirror-Symmetric Decomposable Physical Topologies}\label{sec:polynomial-phy}
The symmetric matrix decomposition only answers how to construct a complete configuration algorithm, at conditions for physical topology shape: The two sub-demands $A$ and $A^{\mathrm T}$ must each be mapped to an OCS configuration on the physical topology. This requires the physical topology itself can be similarly split into two sub-topologies $P_1$ and $P_2$ with mirrored connectivity: if a unidirectional link from port $p$ to port $q$ can be established on $P_1$, the reverse link from the corresponding port of $q$ to the corresponding port of $p$ must be establishable on $P_2$. We call this property \textbf{Mirror-Symmetric Decomposable}, and give a full formal definition in the technical report.

Under this condition, the decoupled optimality is guaranteed by the the same structural property of the algorithm, thus achieving quality lossless decoupling. Our algorithm works as follows: \textcircled{1} Decompose $L$ into $A$ and $A^{\mathrm T}$ using Th.~\ref{lem:matrix_decomp}. \textcircled{2} Solve a configuration realizing $A$ on $P_1$, then reverse every connection; the mirror property yields a configuration realizing $A^{\mathrm T}$ on $P_2$. \textcircled{3} Merging the two gives a configuration realizing the complete $L$. Every step is a polynomial-time combinatorial procedure, and completeness is ensured, with overall complexity $O\left(P^4\log P\right)$ (proved in the technical report).

\begin{theorem}[Mirror Symmetric $\Rightarrow$ Decoupling Optimal]\label{lem:polynomial-solvable}
Mirror-Symmetric Decomposable physical topologies are decoupling optimal, and admits a polynomial complexity OCS configuration algorithm.
\end{theorem}

The full constructive proof is in the technical report.


\subsection{Dual-Link Uniform Wiring (\textnormal{\textsf{2-Uniform}})}\label{sec:2-uniform}

\begin{figure}[!htbp]
    \begin{subfigure}[b]{\linewidth}
     \centering
     \includegraphics[width=\textwidth]{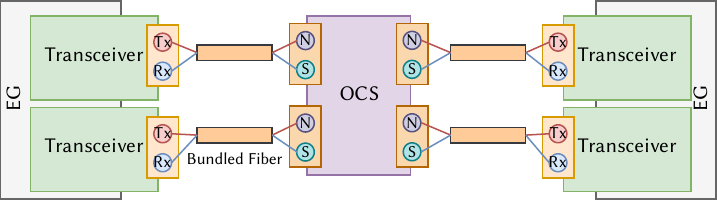}
     \caption{Dual-link Uniform Wiring}\label{fig:dual}
     \end{subfigure}
     
     \begin{subfigure}[b]{\linewidth}
     \centering
     \includegraphics[width=\textwidth]{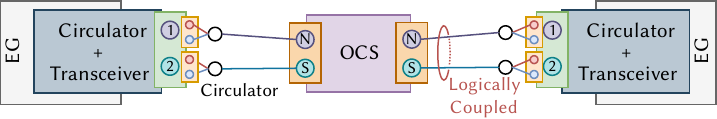}

     \caption{Slim Dual-link Uniform Wiring}\label{fig:2s-uniform}
     \end{subfigure}
     \begin{subfigure}[b]{\linewidth}
     \centering
     \includegraphics[width=\textwidth]{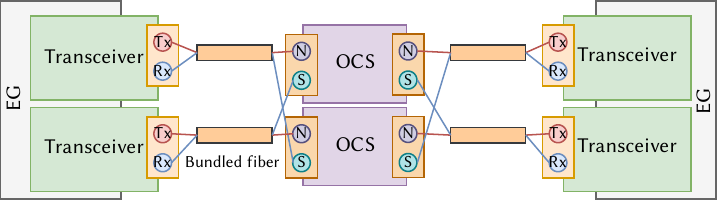}

     \caption{Cross Wiring}\label{fig:itv}
     \end{subfigure}
     \caption{Polynomial Solvable Physical Topologies}
\end{figure}

Now we discuss how to apply Th. \ref{lem:polynomial-solvable} to design decoupling optimal physical topologies. One solution is duplicating a single-direction \textsf{Uniform} wiring scheme, connecting each EG and OCS with two links. 

This is a design adopted by Google~\cite{zu2024resiliency} and Nvidia~\cite{Patronas:25}, while specific optimality-wise reason behind this wiring choice is not articulated in the existing literature and deployments. Nvidia demonstrates with dual-link \textsf{2-Uniform}-style deployment while decisive consideration remains enclosed, while Google motivates circulators as a cost-saving measure and still solves ToE with ILP. We show that such a wiring is in fact \textit{mirror-symmetric} and decoupling optimal, applicable with our polynomial algorithm.

\textsf{2-Uniform} necessitates two links from each EG per OCS. As illustrated in fig. \ref{fig:dual}, with total $2 \times \#_{\text{OCS}}$ ports per EG~\cite{lightcounting_2025_optics}, the ingress and egress of $2k$th, $2k+1$th ports are connected to $k$th OCS. In this case, mirror symmetry is achieved, by dividing the physical topology into two sub-topologies: \textcircled{1} $2k$'s Tx and $2k +1$'s Rx, and \textcircled{2} $2k$'s Rx and $2k+1$'s Tx. The two topologies are identical and each forms a single direction \textsf{Uniform} sub-topology. 

The only concern is the scalability. Here we briefly discuss the scalability of physical topologies.
\begin{sidenote}
    \sfparagraph{Cluster scalability.}

For partial scale and progressive cluster deployment, It is common that multiple Tx/Rx pairs of links are connected between a single EG and OCS to fulfill demands and capable for future expansions. We denote such pair count as $\psi$. When each EG is equipped with $K_{\text{EG}}$ OCS-facing ports and each OCS provides $K_{\text{OCS}}$ EG-facing ports, we quantify the cluster scalability by the \emph{constraint-free maximum connectivity}:
\begin{align}
\text{MC} = \frac{K_{\text{OCS}}}{\psi} \times K_{\text{EG}}.
\end{align}
A cluster cannot scale over the \text{MC} without suffering from partial connectivity over the full interconnect.
\end{sidenote}

For \textsf{Uniform}, the minimum $\psi$ can be 1, while \textsf{2-Uniform} requires $\psi_{\min} = 2$, which reduces the maximum connectivity by half with same type of OCS and EG organization. This is acceptable in conventional deployments built from high-radix 3D-MEMS OCSes, but it becomes problematic with the fast-reconfiguring silicon-photonic switches discussed in \S\ref{sec:suboptimal}: these devices typically offer far fewer ports compared with MEMS-OCS, and the higher radix needed to absorb the halving is not obtainable.

\subsection{Slim Dual-link Uniform Wiring (\textnormal{\normalsize\textsf{2S-Uniform}})}\label{App:slim-dual}

Circulators are optical devices to enable duplex transmission over a single fiber, and utilized in Google datacenters to increase the scale. In this case, a single OCS port pair can connect to two EG's TX and RX pairs with it's ingress and egress port, leading to a halved effective $\psi$. With this device, a slim version of \textsf{2-Uniform} can be implemented with same maximum scale of \textsf{Uniform} as in fig. \ref{fig:2s-uniform}, while preserving mirror-symmetry, thus decoupled optimality.

Though circulator provides the symmetry over EG side, the OCS still have constraint that only connect ingress ports to egress ports (denoted as N-S regions and ports, like magnets in fig. \ref{fig:2s-uniform}), thus minimum 2 link is required in maximum setup. When decomposing the topology, logically a pair of bi-directional links be coupled, and split into two sub-topologies just like \textsf{2-Uniform}: \textcircled{1} the first link's TX paired with second link's RX, and \textcircled{2} first link's RX paired with second link's TX.



    



\sfparagraph{Circulator Deployment.} Circulators have been successfully deployed in Google's datacenters~\cite{poutievski2022jupiter}. However, the deployment faces extra problem compared to regular OCSes.
\begin{itemize}
    \item \textbf{Insertion loss.} Circulators at both ends accumulates roughly 2-4\,dB of extra insertion loss, and the devices themselves carry a non-trivial cost, even comparable to the per-port cost of silicon-photonic OCSes.
    \item \textbf{Crosstalk.} Uni-direction optical systems are not optimized for return loss, i.e. the reflection when crossing optical medium interfaces. When transmitting in duplex. Return loss will cause crosstalk on the other direction and affect signal integrity. Circulator itself can also introduce crosstalk.
\end{itemize}
In design and deployment with Silicon-photonic OCSes, fulfilling insertion loss budgets under -5\,dB and crosstalk budgets under -30\,dB are engineering challenging and already beyond current semiconductor manufacturing and bonding. The budget are getting tighter with optical transceiver upgrade,  \emph{sustaining aggressive return loss performance becomes challenging across such a wide band of operation,} quoted in Google's report~\cite{urata2022missionapollo}. Utilizing circulators constrains the joint choice of transceivers, power compensators, OCS technologies, and topologies. We detail these issues in the technical report.

\subsection{Cross Wiring (\textnormal{\textsf{CW}})} 

Since \textsf{2-Uniform} reduces cluster scalability by half with $\psi = 2$, and \textsf{2S-Uniform} requires extra circulator to work, we propose a new wiring method \textit{Cross Wiring}, which reaches mirror-symmetry with no extra components at maximum scale.

Observing the sub-topology decomposition in \textsf{Uniform}, we need to pair Tx and Rx of two EG ports in a crossing manner. We notice that OCS itself is a link-transparent device that does not require a dual-connectivity: It's not constrained that TX and RX of same EG port must connect to the same OCS. We propose Cross Wiring, by allocating two port pairs onto two OCSes, as in fig. \ref{fig:itv}:
\begin{itemize}[itemindent=0mm, leftmargin=3mm]
\item The TX of the $2k$-th EG port and the RX of the $2k+1$-th EG port are connected to the $2k$-th OCS.
\item The TX of the $2k+1$-th EG port and the RX of the $2k$-th EG port are connected to the $2k+1$-th OCS.
\end{itemize}

The mirror symmetry is achieved that for one sub-topology having port EG A's $2x$ TX connecting to EG B's $2x+1$ RX via OCS $2x$, we have a one-to-one mapped EG B's $2x$ TX connecting to EG A's $2x+1$ RX via OCS $2x+1$, and vice versa. If sub-topologies containing even OCSes can realize a logical sub-topology $A$, its complement containing odd OCSes can realize its transpose $A^T$.

Compared with \textsf{2-Uniform}, implementing the same scale requires using OCSes having doubled ports, and with same OCSes, the scale of \textsf{CW} is doubled, reaching the same scale as \textsf{Uniform}. In fact, the same ToE demand-to-logical-topology algorithm can be applied with \textsf{CW}, with decoupling optimality guaranteeing its solution quality and the polynomial reconfiguration algorithm guaranteeing its efficiency. As we'll discuss in a case study, transforming a \textsf{Uniform} topology to \textsf{CW} requires only wire changes, and gains benefits in quality and efficiency at no cost.

Besides \textsf{Uniform} clusters, \textsf{CW} is also applicable to various topologies with Clos/Rail-shaped OCS connections like TPUv4, where OCS are connected in Clos-shape per dimension in the 3D-Torus structure, with compatibility to existing resource scheduling paradigm\cite{zu2024resiliency}.

\section{Case Study: Deploying Cross Wiring}\label{sec:case_study}

\sfparagraph{Converting Uniform to Cross Wiring.}

One advantage of Cross Wiring is that it can be retrofitted from existing \textsf{Uniform} clusters. By merely modifying the wiring scheme, connections with the same scale and decoupling optimality can be fully realized, without requiring replacement of OCS or EG, and incurring no additional components or costs. A rewiring theme is demonstrated in fig. \ref{fig:wiring_exp}: First, group OCSes in pairs. For each pair of links connecting from the same EG to these two OCSes, rewiring is simply done by plug out two RX fiber and swap them.

\textbf{Rewiring Steps.} For the N-S partitioned OCS shown in the figure, loose connectors are typically adopted given the distance of N-S regions, which enables direct reconfiguration. For dense-connected OCSes employing FC-Duplex connectors (e.g. Calient S320), the N-S ports are arranged interleaved to facilitate direct interfacing with duplex ports. In this scenario, rewiring can be done by removing the coupling latches at the front of the connectors, and gently separating along the prefabricated split line of the zip-cord cable to convert them into loose connector shape, as shown in figure \ref{fig:wiring}.

\textbf{Engineering Feasibility.} A natural concern about implementing Cross Wiring is that insertion loss may drift unevenly across the two OCSes. This is not an issue in practice: even within a single OCS chip, switching channels can differ in insertion loss \cite{11573381}, and the two directions of a bidirectional link traverse two independent unidirectional optical paths, each subject only to the power budget of its own module, so the two sides need not match. Devices are selected against the OCS specification during hardware characterization. Our two-MEMS-OCS testbed confirms this and poses no engineering obstacle.

However, application with MPO optical connectors would be difficult. As in fig. \ref{fig:mpo_wiring}, MPO bundles both TX and RX direction in a single fiber, and extra MPO-LC converter is required to perform.

\begin{figure}[!htbp]
    \centering
    \begin{subfigure}[b]{\linewidth}
        \centering
       \includegraphics[width=0.75\linewidth]{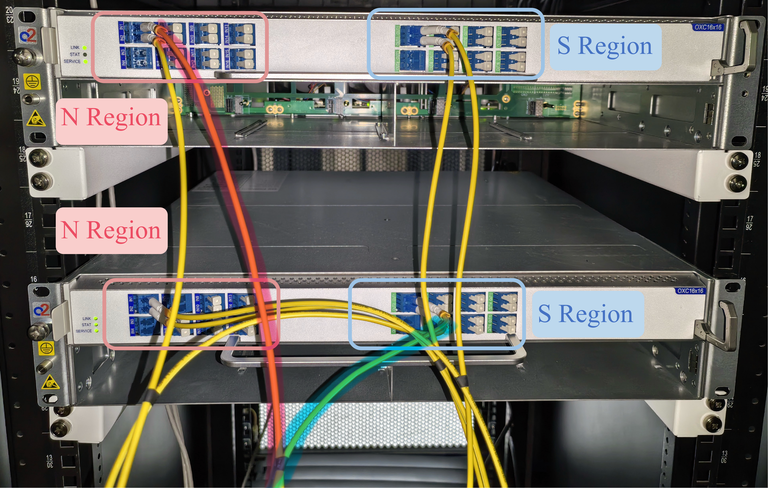}

       \textit{\small One Tx/Rx pair is marked {\color[HTML]{f66936}Red}/{\color[HTML]{37c261}Green} for viewing}
       \caption{Rewiring Example with 2 OCSes for Cross Wiring}\label{fig:wiring_exp}
    \end{subfigure}

    \begin{subfigure}[b]{0.49\linewidth}
        \centering
        \includegraphics[width=0.75\linewidth]{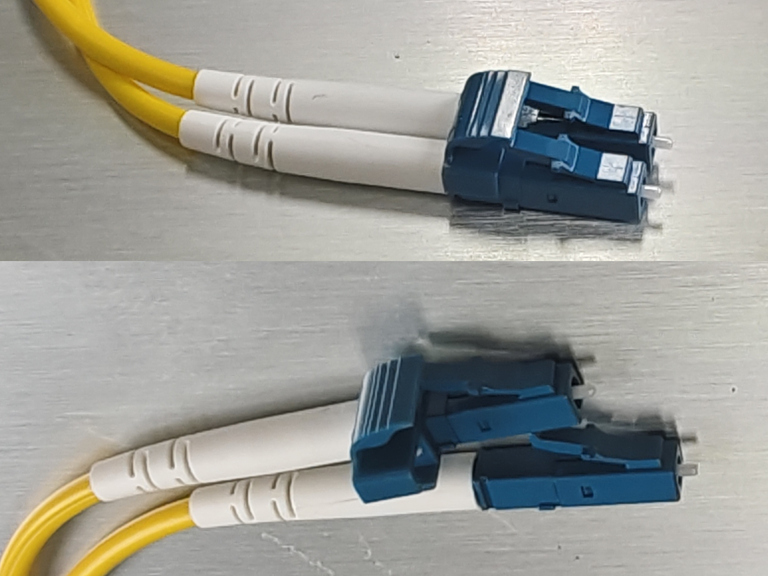}

        \caption{Splitting LC Connectors}\label{fig:wiring}
    \end{subfigure}
    \begin{subfigure}[b]{0.49\linewidth}
        \centering
        \includegraphics[width=0.75\linewidth]{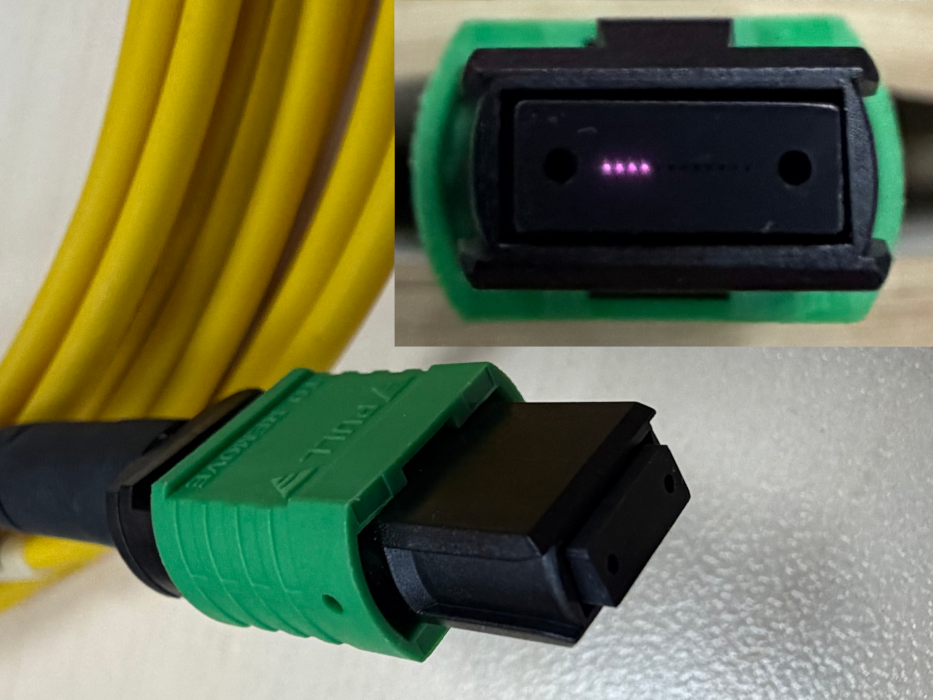}

        \caption{MPO Connectors}\label{fig:mpo_wiring}
    \end{subfigure}

    \caption{Converting \textsf{Uniform} to \textsf{Cross}}
\end{figure}

\noindent\textbf{Discussion:} Most OCS front panels use LC interfaces because an MPO connector bundles many fibers into a single interface, its characteristics limit the port and scale scalability of the OCS. This is, however, an engineering issue rather than a fundamental obstacle. In our technical report, we further detail how Cross Wiring can optimize OCS networking built on MPO interfaces.



\section{Fault Tolerance}\label{sec:discussion}

Fault Tolerance is a critical design criterion for ToE. We analyze two typical failure modes: single port failure (encompassing OCS port, transceiver port, and their connected link failures) and OCS device failure.

\textbf{Port Failure} Single port failure disables a pair of N and S ports. Under Single-link Uniform Wiring, the N and S ports reside on the same OCS; under Cross Wiring, they are distributed across paired OCSes. Given that the N-ports and S-ports operate independently and can be freely paired, Cross Wiring with paired OCSes for routing optical signals presents an identical failure impact perimeter to \textsf{Uniform}.

\textbf{OCS Failure} Under Cross Wiring, a faulty OCS may affect its paired counterpart. Nevertheless, the Mean-Time-Between-Failure (MTBF) of OCSes exceeds 30 years \cite{10.1145/3314148.3314350}. Furthermore, Cross Wiring matches the communication capacity of Dual-link Uniform Wiring using two OCSes, each with \emph{half} the port count. In practice, one OCS can be logically partitioned into two sub-OCSes. Thus, Cross Wiring actually has the identical failure impact perimeter as \textsf{2-Uniform}.

\noindent\textbf{Discussion:} Logical links between EGs are inherently bidirectional, so failures take down whole bidirectional links, affecting neither the \emph{symmetry decompose} algorithm, nor the \emph{mirror-symmetry} of CW physical topology. Some designs~\cite{2023TPU} already bypass failed nodes during logical topology design; others can still exploit the properties of CW to decompose the problem: Unaffected OCSes can be treated as one large logical OCS to shrink the problem~\cite{zhao2019minimal}; excessive demands are then allocated via greedy or ILP over the residual capacity of fail-affected OCSes and this logical OCS. The healthy OCS sub-network then undergoes quality-loss-free ToE per \S~\ref{sec:phoscore}, confining the quality loss to the failed capacity.

Fault tolerance can also be addressed in OCS configuration, as discussed in the technical report.

\section{Micro Benchmark}\label{sec:evn_log}




In this section, we study how physical topology design affects ToE algorithms. End-to-end job performance depends on several factors, including the routing strategy and the design of training frameworks. We want to separate out the effect of the physical topology itself. So we limit our evaluation to ToE solving. We generate 100 logical topology sequences, then run ToE on these sequences one by one.

\noindent\textbf{Datasets.} We use an open-source trace of a multi-tenant cluster~\cite{AcmeTrace2024}. We select LLM training jobs using at least 192 GPUs, grouped into six sizes: 192, 256, 512, 640, 768, and 1024. We mix these sizes in proportion to their trace frequency, packing jobs to exactly fill the cluster. The cluster consists of EGroups, each with 32 OCS-facing ports, connected to OCSes via either \emph{Uniform} Wiring or \emph{Cross Wiring}. We simulate 100 steps; at each step, 25\% of jobs depart and new jobs immediately fill the freed EGroups, modeling cluster churn. Job traffic demands follow the communication pattern.

\subsection{Decoupling Optimality Guarantees Both Solution Quality and Solving Speed}\label{exp:given_logo}

\noindent\textbf{Testbed Setup.} We use a 2048-GPU cluster (64 EGroups) to study how physical topology design impacts ToE. We use MLU as optimization objective, applied ABSM from ATRO~\cite{mao2025atro}, an efficient (averaging 140\,$\mu$s per step) SotA logical-topology solver that is provably MLU-optimal. With this optimality, the MLU is also considered as the lower bound on ideal topologies.


\noindent\textbf{Compared Approaches.} We compare joint and decoupled approaches on both physical topologies. Under \textsf{Uniform}, we evaluate three methods: \textcircled{1}Joint, the joint method that simultaneously solves the logical topology and OCS configuration with ILP; \textcircled{2}ATRO+ILP, a decoupled approach where ATRO generates the logical topology and an ILP computes the OCS configuration; \textcircled{3} ATRO+ILP\textsuperscript{O}, an oracle variant that selects the one with the best MLU among ILP optimized physical topologies. Under Cross Wiring, we evaluate two methods: \textcircled{4} Joint; and \textcircled{5} ATRO+MCF, a two-step decoupled approach in which ATRO runs first, then the OCS configuration is derived via our MCF-based polynomial algorithm. We use raw ATRO as the baseline. We use Gurobi \cite{achterberg2019s} for solving ILP.
\vspace{-2ex}
\begin{figure}[!htbp]
    \centering
    \includegraphics{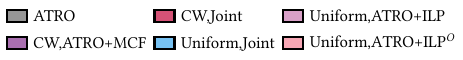}
    \begin{subfigure}[b]{0.49\linewidth}
        \centering
        \includegraphics{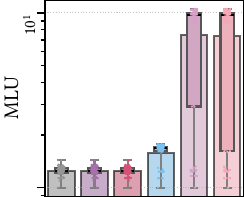}
        \caption{MLU}
        \label{fig:mt_mlu_2k}
    \end{subfigure}
    \begin{subfigure}[b]{0.49\linewidth}
        \centering
        \includegraphics{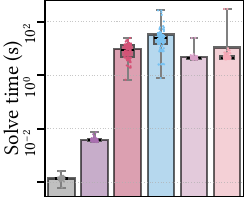}
        \caption{Per-step solve time}
        \label{fig:mt_solve_time_2k}
    \end{subfigure}
    \caption{Cross Wiring is Fast and Efficient}
    \label{fig:mt_decouple_2k}
\end{figure}

\textbf{Decoupling Optimality preserves Quality.} Fig.~\ref{fig:mt_mlu_2k} shows MLU for these methods and the ATRO bound. Under Cross Wiring, decoupled ATRO+MCF and Joint achieve bitwise-identical MLU, matching the ATRO bound. Thus, under Cross Wiring, logical topology design and OCS configuration solving can be fully decoupled with no loss of optimality, while \emph{Uniform} suffers from this quality loss. Relative to Cross Wiring, Joint's mean MLU is 26.5\% worse, while ATRO+ILP\textsuperscript{O} and ATRO+ILP are 492.8\% and 499.3\% worse. Even the oracle solver still $\sim$4.7$\times$ worse than the joint method, demonstrating the quality loss as a direct affect of losing decoupling optimality.

\textbf{Decoupling Optimality improves Speed.} Fig.~\ref{fig:mt_solve_time_2k} shows the per-step solving time. Our proposed algorithm over Cross Wiring averages 3.9 ms. For ILP solving over Cross Wiring (9.34 s), it gains 3.6$\times$ speedup compared to the same solver over Uniform topology, which averages 33.6 s per step, $\sim$8400$\times$ slower than our algorithm (including time cost of ATRO). This shows regular solvers can also benefit from Cross Wiring with Decoupling Optimality.

\subsection{Decoupling Optimality Simplifies Quality Efficiency Balancing}

\noindent\textbf{Compared Approaches.} Due to complexity of full ILP solving, many approximation OCS configuration algorithms have been proposed. We compare \textcircled{1} BvN: Birkhoff-von Neumann method~\cite{0On,832560,9509349,10.1145/2716281.2836126}, which decomposes the demand matrix into permutation matrices, each representing an OCS configuration; \textcircled{2} MixNet~\cite{Liao_2025}, greedily pairs ports under a fairness rule. \textcircled{3} DFS: Simple depth-first-search with branch pruning. These fast methods sacrifices quality in order to gain different level of efficiency, to meet the reconfiguration timing budget.

\noindent\textbf{Benchmark Design.} To assess OCS configuration algorithms, we fix the logical topology generator to ATRO and vary only the subsequent OCS configuration algorithm. We measure performance with LTCR (Eq.~\eqref{eq:LTCR}), which quantifies how much of the logical topology is realized. We evaluate two physical topologies, Uniform and Cross Wiring, at scales 2k, 4k, 8k, and 16k, using the same 100-step sequence for all scales.

\begin{figure}[!htbp]
    \centering
    \includegraphics{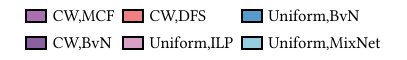}
    
    \begin{subfigure}[b]{0.49\linewidth}
        \centering
        \includegraphics{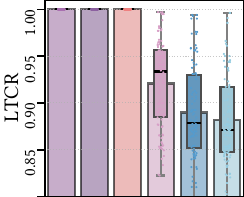}
        \caption{4096 GPUs}
        \label{fig:mt_rr_4k}
    \end{subfigure}
    \begin{subfigure}[b]{0.49\linewidth}
        \centering
        \includegraphics{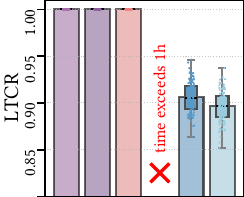}
        \caption{16384 GPUs}
        \label{fig:mt_rr_16k}
    \end{subfigure}
    {\footnotesize\emph{Uniform+ILP timeout for 16k. Results of 2k/8k are reported in the technical report.}}
    \caption{Realization rate of OCS configuration algorithms.}
    \label{fig:mt_realization_rate}
\end{figure}

\noindent\textbf{Good Topology Elevates Simple Methods.} Fig.~\ref{fig:mt_realization_rate} shows the LTCR using different OCS configuration algorithms. Under Cross Wiring, methods behave identically. All three give LTCR = 1 at every one of the four scales. In other words, every demand pair in the logical topology is realized in full. With a good topology, even OCS configuration algorithms that are designed to trade quality for speed can reach a better configuration.

\noindent\textbf{Suboptimal Topology Penalize Efficiency.} Under \textsf{Uniform}, neither of the two well-designed heuristic algorithms can meet the bandwidth demands set by the logical topology, and simple methods like MixNet (a greedy one) are penalized with greater quality loss. Without decoupling optimality, the quality loss is inevitable, due to harder constraints created by the wiring structure. Simple and efficient methods without knowledgeable heuristics and complex constraints will have greater quality losses. This complicates the algorithm designers, the efficiency and quality needs to be carefully balanced, and subject to change when the cluster scale changes and workload shifts.


\section{The Efficiency-Quality Dilemma}\label{sec:layer_eval}

Silicon-photonic OCSes \cite{TODO_neye, ipronics_web, 11573381} now reconfigure in microseconds, enabling fabrics to follow rapidly changing traffic. Recent OCS-oriented LLM systems such as \textit{MixNet}~\cite{Liao_2025}, \textit{Opus}~\cite{10.1145/3789240.3829187}, and \textit{RailX}~\cite{ma2025railx} already reconfigure at collective-communication granularity. A reconfiguration cycle combines the OCS step, optical-link restoration, and ToE solving. Engineering optimization compresses the hardware steps to $\sim$200\,$\mu$s, and ToE solving taking milliseconds to seconds becomes the new bottleneck.

As a example, consider reconfiguring the OCS for each MoE dispatch operator. The system predicts the next layer's token distribution at the start of the current forward dispatch and launches ToE calculation, overlapping solving with all computation and communication before the next dispatch, and Solving time within the computation window can be fully hidden. The reconfigured OCS will then serve the communication~\cite{Liao_2025}. As elaborated in \S\ref{sec:uniform}, when computation window cannot fully cover, iteration time will \textbf{inflate}. Under \textsf{Uniform}, exact solving takes tens of milliseconds and no polynomial-time complete algorithm exists. MixNet and BvN meet the deadline by sacrificing the quality, which increases contention. This causes the efficiency and quality not a trade-off, but a dilemma since sacrificing either can affects the communication and throughput.

To evaluate large scale behavior, we profile operators on a three-node H800 cluster with CX7 NICs to calibrate Astra-sim~\cite{won2023astrasim2}, and simulate 2048-GPU pipelines with per-layer solving times from the actual solvers. We compare MCF on CW, MixNet~\cite{Liao_2025} and BvN~\cite{0On} on \textsf{Uniform}, and exact ILP on both, sharing ATRO as the logical topology solver~\cite{mao2025atro} to isolate. We evaluate 5 MoE models (DeepSeek-V3, Qwen3-235B, GLM-4.5-Air, gpt-oss, and DSv2-Lite) in Fig.~\ref{fig:iter_model_gap}.
\vspace{-2ex}
\begin{figure}[!htb]
    \centering
    \vspace{-0.5ex}
    \includegraphics{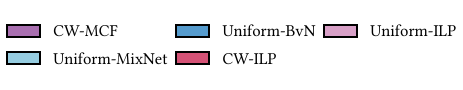}
    \vspace{-2ex}
    
\includegraphics{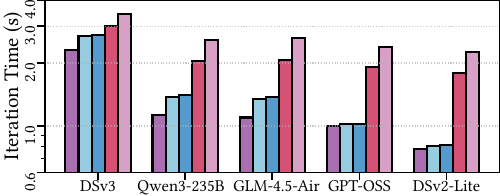}
    \caption{The Efficiency-Quality Dilemma}
    \label{fig:iter_model_gap}
\end{figure}

\sfparagraph{Slow solvers affects performance.} To compare the performance, in our evaluation setup, LTCR of ILP for both topologies can generally reach 1. CW-ILP and Uniform-ILP take 31.56 and 49.55\,ms per layer.  Though CW cuts solving time by 36.3\% with symmetry, but tens of milliseconds remain unhidden, and the penalty across the five models grows from +48.3\% (DeepSeek-V3) to +190.9\% (DSv2-Lite) for lighter models with thinner overlap windows.

\sfparagraph{Decoupling Optimality saves quality and efficiency.} For DeepSeek-V3 (TP=2, PP=16, EP=64), MixNet, BvN, and MCF take 1.19, 1.22, and 0.14\,ms to solve. Compared to ILP exact solvers, approximation improves iteration time even with less fit logical topologies since they don't stall the pipeline with \textsf{Uniform}. Our MCF exact solving algorithm is efficient, as under \textsf{Uniform}, MixNet and BvN still run 16.3\% and 18.2\% slower than CW-MCF, leading to minimal iteration time among comparison.

\textbf{Sensitivity.} Sensitivity analyses over bandwidth, prediction accuracy, demand distribution, batch size, and scheduling confirm this pattern. With decoupling optimality, both quality and efficiency can be attained.

\section{The Importance of Topology Quality}\label{section:testbed}
\begin{sloppypar}

\subsection{Testbed Evaluation}

For fast switching OCSes, our polynomial algorithm and decoupled optimal physical topology helps to break the quality and efficiency dilemma, providing quality lossless and computation efficiency. For widely deployed clusters with millisecond-level MEMS-OCSes, Cross Wiring also provides benefits of solution quality. Deployments like TPUv4~\cite{2023TPU} and TopoOPT~\cite{wang2022topoopt} reconfigure the optical fabric at job granularity to allocate interconnect among multiple tenants.  A job's communication performance still depends on whether its logical topology is realized exactly, so these clusters equally need a Decoupling-Optimal physical topology. We quantify the end-to-end impact of quality loss on a physical testbed.


\textbf{Setup and Compared Topologies.} Our MEMS-OCS-based testbed has 128 Ascend 910A NPUs across 16 servers with 100\,Gbps RoCE cross-server links, and ToE is solved once globally for each arriving job. We compare three topologies: \textcircled{1} CW, \textcircled{2} the commonly used Uniform, and \textcircled{3} Best, a large EPS that interconnects all NPUs and gives a contention-free upper bound. We pre-plan flow paths with source UDP port hashing to minimize contention \cite{HAN2025111285}, so the measured gap reflects only logical-topology realization, namely the impact of Decoupling Optimality.

\begin{figure}[!htbp]
\vspace{-2ex}
     \begin{subfigure}[b]{0.49\linewidth}
     \includegraphics[width=\linewidth]{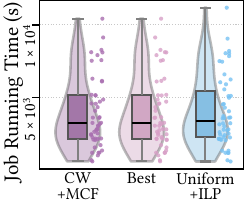}
    \caption{Loss of Decoupling Optimality degrades training throughput.}
    \label{fig:pdf_of_JRT}
     \end{subfigure}
     \begin{subfigure}[b]{0.49\linewidth}
     \includegraphics{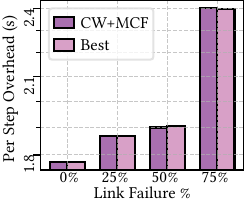}
    \caption{ Cross Wiring tolerates link failures.}
    \label{fig:link_failure_ratio}
     \end{subfigure}
     \caption{Results on the 128-NPU testbed.}
\end{figure}

We compare 5 models run on 96 NPUs (TP$=8$, PP$=2$, DP$=6$, and EP$=2$ for Pangu-$\alpha$ and GPT2). We replay a 48-hour, 50-job trace on MindFormers 1.0.0~\cite{tong2021study}, covering Llama-7B, Llama2-7B/13B, Pangu-$\alpha$-6B, and GPT2-13B. Job size $N$ is drawn from $\{16, 32, 64, 96, 128\}$. TP/EP traffic stays within one Pod, with TP$=8$, PP drawn from $\{1, 2, \ldots, N/8\}$, DP$=N/(\mathrm{PP}\times\mathrm{TP})$, and EP$=2$ for Pangu-$\alpha$ and GPT2. \textsf{Uniform} cannot fully realize their logical topologies, increasing iteration time by up to 39.5\% compared with CW. 

Fig.~\ref{fig:pdf_of_JRT} shows that CW reduces average job running time by up to 22.1\% while staying within 1\% of Best. On a 96-NPU Llama2-13B job with port-shutdown failures, Fig.~\ref{fig:link_failure_ratio} shows that CW matches Best in link fault tolerance, consistent with \S\ref{sec:discussion}.

\subsection{Large-Scale Simulation}\label{ml_sim}

\indent\textbf{Setup.} We extend the task-level scenario to large-scale simulation, replaying a real multi-tenant trace on a TPUv4-like OCS cluster, calibrating the simulator against the testbed, and quantifying how incomplete topology realization affects job running time. Each EG contains 64 GPUs with 96 $\times$ 200\,Gbps uplinks to the OCS layer. With 128-port OCSes, CW interconnects 8192 GPUs, whereas TPUv4 interconnects 4096. We evaluate CW and Uniform against a non-oversubscribed three-tier electrical Clos (\emph{clos}) and an unlimited-port single-switch upper bound (\emph{Best}). Routing follows the testbed; ECMP is reported in the sensitivity analysis.

\textbf{Dataset.} Our trace combines public TPUv4 information with the open-source AcmeTrace~\cite{AcmeTrace2024}. Job sizes and class shares come from TPUv4 and runtimes from AcmeTrace. Runtimes are highly skewed: 50\% of jobs finish within 7 minutes and about 1\% run over 69 hours, reflecting repeated model tuning before formal pre-training. ToE must therefore stay fast, because a solve over 5 minutes hurts user experience, and a long job placed on a wrong topology suffers a severe training-time tail. Job arrivals follow a Poisson process with mean interval $\beta$. Cluster load is $\rho=\sum_n(n\lambda_nT_n)/\text{GPU}_{\text{Num}}$, the average GPU demand over total GPUs by Little's law. We use $\rho=0.5$ by default and sweep it from 0.46 to 0.60 (details in the technical report).

\textbf{Simulator.} NS3~\cite{6663585} and NetBench~\cite{Kassing2016Netbench} take days per trace hour and cannot reconfigure the network during communication. We instead build Phos, an event-driven flow-level simulator calibrated against the testbed for accurate job running time. The code will be open sourced in our technical report. MLU is uninformative for ML traffic, whereas Phos reports end-to-end job outcomes and is better suited to ToE research.

Joint modeling for communication demands is difficult for task-level ToE: a job spans many operators across multiple EGs, and capturing their link allocation at all times would require far more variables than the operator-level model in \S\ref{sec:layer_eval}. We therefore use a decoupled approach: extract the execution DAG before launch, exploit link reuse across operators, design the logical topology with an MCF method and spare-link robustness, and take ATRO as a reference.

\sfparagraph{Quality loss affects communication and running time.} At 8k scale, among jobs that cross EGs, namely those using over 128 GPUs (37.7\% of all jobs), CW-MCF cuts average communication time by 5.9\% and 26.7\% versus Uniform-BvN and Uniform-MixNet. Under Uniform, 23.3\% and 36.6\% of jobs see over 10\% more communication time, and 11.7\% and 20.4\% see over 3\% more running time.

Once this quality loss is eliminated, an OCS cluster matches an electrical one at lower cost. CW-MCF's average JRT is only 0.13\% above clos, consistent with prior NVIDIA analysis~\cite{Patronas:25}, and it stays close to clos from 2K to 16K (Fig.~\ref{Avg.JRT_vs_scale}).

Decoupling also lets the two stages evolve independently. Replacing logical-topology design with the SotA ATRO raises 8K JRT by 5.5\% and leaves it 5.6\% behind clos, because MLU misses the temporal dependencies and link time-sharing across communication domains. MixNet-joint, applying MixNet's strategy to both logical-topology design and OCS configuration, targets communication-domain-level rather than task-level switching also performs not well. Different scenarios need different ToE algorithms and objectives; decoupling lets each stage be replaced independently.

\sfparagraph{Sensitivity analysis.} The conclusions hold when we sweep load $\rho$ from 0.46 to 0.60 at 4K and 8K and replace pre-planned routing with ECMP at 8K. Per-scale communication time and other details are in the technical report.

\begin{figure}[t]
    \centering
    \vspace{-2ex}
    \includegraphics{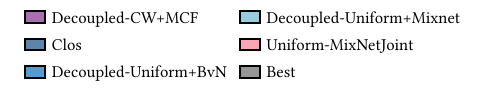}
    \begin{subfigure}[b]{0.49\linewidth}
     \includegraphics{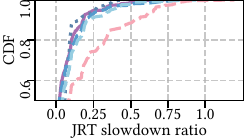}
     \caption{Slowdown ratio of JRT at 8K scale relative to \emph{Best}}\label{case1}
    \end{subfigure}
    \begin{subfigure}[b]{0.49\linewidth}
     \includegraphics{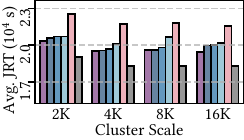}
     \caption{Average JRT across cluster scales}\label{Avg.JRT_vs_scale}
    \end{subfigure}
\caption{Evaluating jobs using more than 128 GPUs.}
\vspace{-2ex}
\end{figure}

\end{sloppypar}

\section{Conclusion} 
ToE is an active topic in OCS cluster research, and decoupled formulations are widely adopted for their flexibility and modeling simplicity. The resulting quality loss, however, has long been viewed as an inherent cost of decoupling. In this paper, we revisit ToE from the perspective of physical topology design and show that the physical topology, rather than decoupling itself, is what determines this quality loss. We propose Cross Wiring, which is decoupling optimal. Testbed measurements and large-scale simulations validate the impact of physical topology on ToE, while the MEMS-OCS testbed demonstrates the engineering feasibility of our design.


\end{sloppypar}

\clearpage
\bibliographystyle{ACM-Reference-Format}
\bibliography{sample-base}

\end{document}